\documentclass[twoside,leqno]{article}

\usepackage[letterpaper]{geometry}

\usepackage{ltexpprt}
\usepackage{hyperref}
\usepackage{framed}
\usepackage{setspace}

\usepackage{booktabs} 
\usepackage{mathtools}
\usepackage{multirow}
\usepackage{amsbsy, amsfonts, amsgen,  amsopn, amssymb, amstext,
 amsxtra, bezier, bbm, color, enumerate, graphicx, latexsym, verbatim,
	pictexwd, supertabular, url, dsfont,leftidx, mathrsfs, mathtools, setspace, bbm}
\usepackage[toc,page,header]{appendix}
\usepackage{minitoc}

\usepackage{algorithm, algpseudocode}
\newcommand{\mb}[1]{\ensuremath{\boldsymbol{#1}}}
\newcommand{\mc}[1]{\mathcal{#1}}

\def\opt{\textsc{OPT}}
\def\alg{\textsc{ALG}}

\def\mbb#1{\mathbb{#1}}

\begin{document}

\title{\Large When Stochastic Rewards Reduce to Deterministic Rewards in Online Bipartite Matching}
\author{Rajan Udwani \thanks{Department of Industrial Engineering and Operations Research, University of California, Berkeley. Part of this work was supported by a Google Research Scholar award.}}

\date{}

\maketitle







\begin{abstract} \small\baselineskip=9pt 	We study the problem of vertex-weighted online bipartite matching with stochastic rewards where matches may fail with some known probability and the decision maker has to adapt to the sequential realization of these outcomes. Recent works have studied several special cases of this problem and it was known that the (randomized) Perturbed Greedy algorithm due to Aggarwal et al.\ (SODA, 2011) achieves the best possible competitive ratio guarantee of $(1-1/e)$ in some cases. We give a simple proof of these results by reducing (special cases of) the stochastic rewards problem to the deterministic setting of online bipartite matching (Karp, Vazirani, Vazirani (STOC, 1990)). More broadly, our approach gives conditions under which it suffices to analyze the competitive ratio of an algorithm for the simpler setting of deterministic rewards in order to obtain a competitive ratio guarantee for stochastic rewards. The simplicity of our approach reveals that the Perturbed Greedy algorithm has a competitive ratio of $(1-1/e)$ even in certain settings with correlated rewards, where no results were previously known. Finally, we show that without any special assumptions, the Perturbed Greedy algorithm has a competitive ratio strictly less than $(1-1/e)$ for vertex-weighted online matching with stochastic rewards.\end{abstract}

\section{Introduction}
In online matching with stochastic rewards, we have a graph $G=(I,T,E)$ where the vertices in $I$, called \emph{ resources}, are known in advance. Vertices in $T$, called \emph{arrivals}, are sequentially revealed one at a time. When a vertex $t\in T$ arrives, every edge $(i,t)\in E$ incident on $t$ is revealed along with the corresponding probability of success $p_{it}$. 
On every new arrival, we must make an irrevocable decision to match the arrival to any one of the available offline neighbors. If arrival $t$ is matched to resource $i$, the match succeeds with probability $p_{it}$ independent of all other events. If the match succeeds, we obtain a reward of $r_i$ (equal to the vertex weight) and resource $i$ is unavailable to future arrivals. If the match fails, arrival $t$ departs, resource $i$ is available to match to future arrivals and we get no reward. The objective is to maximize the expected reward summed over all arrivals. 
Let $\opt(G)$ denote the optimal achievable by an offline algorithm and let $\mc{A}(G)$ the expected value achieved by a (possibly randomized) online algorithm $\mc{A}$. We measure the worst case performance of $\mc{A}$ by evaluating the competitive ratio,  
\[\inf_{G}\frac{\mc{A}(G)}{\opt(G)}.\]
Our goal is to design an algorithm that has the highest possible competitive ratio. This can be quite challenging and a more relaxed goal is to find an algorithm that outperforms 
a na\"ive greedy algorithm that matches each arrival to the resource that maximizes expected reward (or show that no such algorithm exists). 
It can be shown that no online algorithm has competitive ratio better than $O(1)/|I|$ against an \emph{omniscient} offline algorithm that knows the realization of all edges\footnote{From \cite{stochrew}: Consider a single arrival with an edge to $n$ identical resources. Each edge has a probability of success $1/n$. A
	benchmark that knows realizations of rewards attains expected reward at least $(1-1/e)$. However, no online algorithm
	can achieve a reward better than $1/n$.}. Similar to related work \cite{negin,will, stochrew, brubach2, borodin2022prophet,reuse,ms}, we consider a natural alternative (non-omniscient) benchmark that 
knows all edge probabilities and can match arrivals in \emph{any order} but does not know the realization of an edge until after it is included in the matching. The algorithm can adapt to edge realizations and we call this the \emph{non-anticipative} offline benchmark. Note that the benchmark can match each arrival at most once. We refer to \cite{stochrew} for a discussion on other possible benchmarks. 

\textbf{Previous work:}  This problem and its generalizations as well as special cases have various applications ranging from online advertisement \cite{survey} to volunteer matching on non-profit platforms \cite{saban2}. The special case where all edges have success probability of 1 and vertex weights are identical corresponds to the classic online bipartite matching problem introduced by \cite{kvv}, who proposed and analyzed the Ranking algorithm that orders all resources in a random permutation and matches every arrival to the highest ranked neighboring resource that is available. They showed that Ranking attains the highest possible competitive ratio of $(1-1/e)$ for online bipartite matching\footnote{See \cite{baum,goel2,devanur,sosa} for simplified analysis.}.  \cite{deb} initiated the study of the stochastic rewards problem with identical edge probabilities and identical vertex weights. They showed that the Ranking algorithm is 0.534 competitive and proposed a new online algorithm with an improved guarantee of $0.567$ when the probabilities are identical and also vanishingly small. Subsequently, \cite{mehta} gave a 0.534 competitive algorithm for the case of vanishingly small but heterogeneous probabilities. They also showed that the greedy algorithm that matches each arrival to an available resource with the highest expected reward (which equals the highest probability of success), is 0.5 competitive for the general setting. \cite{negin} showed that the greedy algorithm is 0.5 competitive much more broadly, including settings with heterogeneous vertex weights. This is still the best known competitive ratio guarantee for the general stochastic rewards problem.

For settings with heterogeneous vertex weights, \cite{goel} gave a generalization of the Ranking algorithm, called \emph{Perturbed Greedy}, and showed that it achieves the highest possible competitive ratio of $(1-1/e)$ when edge probabilities are 1. 
More recently, \cite{stochrew} showed that a natural generalization of Perturbed Greedy (see Algorithm \ref{rank}) is $(1-1/e)$ competitive when edge probabilities can be decomposed into a product of two probabilities, one corresponding to the resource and the other corresponding to the arrival, i.e., $p_{it}=p_ip_t\,\,\forall (i,t)\in E$. This includes the setting of identical edge probabilities and identical vertex weights (studied by \cite{deb}) as a special case. They also gave a 0.596 competitive algorithm for settings where edge probabilities are vanishingly small. Concurrently, \cite{huang} gave a 0.572 competitive algorithm for vanishingly small edge probabilities against a stronger benchmark. Very recently, \cite{unknown} showed that the Perturbed Greedy algorithm is 0.508 competitive when probabilities are vanishingly small\footnote{They also show an improved competitive ratio of 0.522 for a slightly modified version of Algorithm \ref{rank}.}. 

To summarize, there are two well studied special cases of the stochastic rewards problem. Firstly, the case where probabilities are identical, and more generally, decomposable. Secondly, the case where probabilities are vanishingly small. In this paper, our focus is on the first strand. Note that Perturbed Greedy is the only algorithm that is known to have a competitive ratio strictly better than 0.5 in all well studied special cases of the stochastic rewards problem. 

\begin{algorithm}
	\caption{Perturbed Greedy for Matching with Stochastic Rewards}\label{rank}
\begin{algorithmic}[1]
\State	Let $S = I$ and for every $i\in I$ generate i.i.d.\ r.v.\ $y_i\in U[0,1]$
	\For{\emph{every new arrival } $t$}
\State	\textbf{match} $t$ to	$i^*=\underset{i\mid (i,t)\in E, i\in S}{\arg\max} p_{it}r_i (1-e^{y_i-1})$
\Comment{match greedily using perturbed expected rewards}
		\If {match succeeds} 
		\State \textbf{update} $S=S\backslash\{i^*\}$
\EndIf
\EndFor
\end{algorithmic}
\end{algorithm}
Classic techniques, such as the randomized primal-dual method of \cite{devanur} are very well suited for deterministic rewards but have proven to be of limited use in settings with stochastic rewards. While Perturbed Greedy has the same competitive ratio guarantee for deterministic rewards and for stochastic rewards with decomposable probabilities, the existing analyses of Perturbed Greedy and all other algorithms for stochastic rewards rely on various new technical developments. For example, \cite{stochrew} propose a new (and more sophisticated) sample path based analysis, \cite{huang} show new structural properties and develop a new (and more sophisticated) primal-dual analysis. 

\textbf{Our contributions:} 
We give a simple proof of $(1-1/e)$ competitiveness of Perturbed Greedy for vertex-weighted online matching with identical edge probabilities, and more generally, for decomposable edge probabilities. Our method reduces each instance of the stochastic rewards problem to a distribution over instances with deterministic rewards. The $(1-1/e)$ competitiveness of Perturbed Greedy then follows directly from results in deterministic settings. More broadly, our approach shows that for stochastic rewards with decomposable edge probabilities, it suffices to analyze the competitive ratio of an algorithm for the special case of deterministic rewards in order to obtain a competitive ratio guarantee for stochastic rewards. Further, the simplicity of our approach reveals that the Perturbed Greedy algorithm has a competitive ratio of $(1-1/e)$ even in certain settings with correlated rewards, where no results were previously known. Finally, we show that without any special assumptions, the Perturbed Greedy algorithm has a competitive ratio strictly less than $(1-1/e)$ for vertex-weighted online matching with stochastic rewards. Prior to our work, no upper bound was known on the competitive ratio of Perturbed Greedy for the general setting.

\section{Preliminaries}\label{sec:prelim}
We start by recalling the classic vertex-weighted online bipartite matching problem. We refer to it as the \emph{deterministic rewards} problem.

\textbf{Deterministic rewards setting:} Consider a graph $G=(I,T,E)$ with $n$ resources $i\in I$ that are known in advance and $m$ arrivals $t\in T$ that are sequentially revealed one at a time. W.l.o.g., the vertices in $T$ arrive in the order of their index. Let $N(i)$ denote set of arrivals with an edge to $i$ , i.e., neighbors of $i$. Similarly, let $N(t)$ denote the neighbors of $t$. When a vertex $t\in T$ arrives, all edges $(i,t)\in E$ incident on $t$ are revealed and we must make an irrevocable decision to match the arrival to any one of the available offline neighbors. If arrival $t$ is matched to resource $i$, we get a reward $r_i$ (the weight of vertex $i$) and $i$ is unavailable to future arrivals. The objective is to maximize the total reward summed over all arrivals. We will use the following classic result.
\begin{lemma}[\cite{goel, devanur}]\label{fact1}
	Perturbed Greedy is $(1-1/e)$ competitive for the deterministic rewards setting.
\end{lemma}

\textbf{Stochastic rewards setting:} This is a generalization of the deterministic rewards setting, where matches can fail with some known probability. In particular, when $t$ arrives, all edges incident on $t$ are revealed along with probabilities $p_{it}\,\, \forall i\in N(t)$. If $t$ is matched to $i$, the match succeeds with probability $p_{it}$ independent of all other events and we get a reward $r_i$. If the match fails, $t$ departs, resource $i$ is available to match to future arrivals and we get no reward. When all edge probabilities are equal to 1, we obtain the deterministic rewards setting. 
We are interested in the following special cases of the problem. 
\begin{itemize}
	\item \textbf{Identical probabilities:} 
	$p_{it}=p\,\, \forall (i,t)\in E$. 
	\item \textbf{Decomposable probabilities:} 
	$p_{it}=p_i\, p_t\,\,\forall (i,t)\in E$. This subsumes identical probabilities. 
\end{itemize}
We evaluate competitive ratio against a non-omniscient benchmark that 
knows all edge probabilities and can match arrivals in \emph{any order} but does not know the realization of an edge until after it is included in the matching. The algorithm can adapt to edge realizations and we call this the \emph{non-anticipative} offline benchmark. 
\smallskip

\textbf{$b$-matching:} Introduced by \cite{pruhs}, this is a generalization of the deterministic setting where resources have starting budgets or capacities given by positive integers $\{b_i\}_{i\in I}$. The budget of a resource specifies the number of arrivals it can be matched to. We get a reward of $r_i$ every time we match resource $i$. When $b_i=1\,\, \forall i\in I$, we obtain the deterministic setting specified above. We are particularly interested in the setting where the online algorithm does not know the budgets $\{b_i\}_{i\in I}$ at the beginning, called \textbf{$\mb{b}$-matching with unknown budgets}. In this setting, the budget of a resource is revealed to the online algorithm only after it has been used up, i.e., if resource $i$ is matched to $b_i$ arrivals then the budget of $i$ is revealed to the online algorithm after $i$ is matched for the $b_i$-th time.  The Perturbed Greedy algorithm is well defined for this setting -- we compute a perturbed vertex weight for each resource and match each arrival to a resource with available budget and the highest perturbed weight. We use the following recent result for Perturbed Greedy in this setting. \cite{vazirani} gives a particularly elegant proof of the result. 
\begin{lemma}[\cite{albers, vazirani}]\label{fact2}
	Perturbed Greedy is $(1-1/e)$ competitive for $b$-matching with unknown budgets against offline that knows all budgets in advance.
\end{lemma}

\textbf{Adwords:} Introduced by \cite{msvv}, this is a generalization of the $b$-matching problem where resources make bids on the arrivals they have an edge to. Let $b_{it}$ denote the bid made by $i$ on arrival $t$. The bid is revealed when $t$ arrives and it is the number of units of $i$ that will be used if $i$ is matched to $t$. If $i$ and $t$ are matched, we get a reward of $b_{it}$ subject to the available budget. Let $b_i(t)$ denote the remaining budget of $i$ when $t$ arrives. Matching $i$ to $t$ uses $\min\{b_i(t), b_{it}\}$ of $i$'s remaining budget and we get a reward equal to the number of units that are used (=$\min\{b_i(t), b_{it}\}$). An instance of the $b$-matching problem can be represented as an instance of Adwords by setting the starting budget to $b_ir_i$ and bids to $b_{it}=r_i$ for all $i\in I,\, t\in N(i)$.  In the \textbf{Adwords with unknown budgets} setting, the starting budgets are unknown to the online algorithm and the budget of a resource is revealed only once it is used up.
\smallskip

\textbf{Non-anticipative algorithms:} For the stochastic rewards setting, we say that an algorithm $\mc{A}$ (could be online or offline) is non-anticipative when it does not know the realization (success/failure) of an edge before it is included in the matching. The outcome of an edge is revealed to $\mc{A}$ immediately after it includes the edge in the matching. Note that an any algorithm for stochastic rewards is well defined on instances with deterministic rewards (which are special cases). 

We use $\alg$ to denote the Perturbed Greedy algorithm for stochastic rewards. We use $\opt$ to denote the non-anticipative offline benchmark for stochastic rewards. Given an instance $\nu$ of some setting and an algorithm $\mc{A}$ for the setting, we use $\mc{A}(\nu)$ to denote the expected total reward in the setting with expectation taken over the intrinsic randomness in the (possibly randomized) algorithm as well as the randomness in the instance (such as stochastic rewards). In particular, $\alg(\nu)$ and $\opt(\nu)$ denote the expected total reward on instance $\nu$ of \alg\ and \opt\ respectively. Finally, we define the notion of a value preserving distribution. The notion plays an important role in our main results.
\smallskip

\noindent \textbf{Value preserving distribution:} Given an instance $\nu$ of the stochastic rewards problem and a distribution $\mc{D}_\nu$ over instances of the stochastic rewards problem, we say that $\mc{D}_\nu$ is value preserving if for every non-anticipative algorithm $\mathcal{A}$,
\[\mc{A}(\nu)=\mbb{E}_{u\sim \mc{D}_\nu}[\mc{A}(u)].\]


\section{Main Results}

\begin{theorem}\label{meta}
	Consider a set $\mc{V}$ of instances of the stochastic rewards problem and a set of value preserving probability distributions $\{\mc{D}_{\nu}\mid \nu\in\mc{V}\}$. 
We have that,
	\[\inf_{\nu\in\mc{V}}\frac{\alg(\nu)}{\opt(\nu)}\geq \inf_{u\in \underset{\nu\in\mc{V}}{\bigcup} supp(\mc{D}_\nu)}\frac{\alg(u)}{\opt(u)},\]
	here $supp(\mc{D}_\nu)$ represents the support of distribution $\mc{D}_\nu$. 
\end{theorem}
\begin{proof}
	Using the fact that \alg\ and \opt\ are non-anticipative, we have,
	\[\frac{\alg(\nu)}{\opt(\nu)}=\frac{\mbb{E}_{u\sim \mc{D}_\nu}[\alg(u)]}{\mbb{E}_{u\sim \mc{D}_\nu}[\opt(u)]} \qquad \forall \nu\in\mc{V}.\]
	Let $\alpha$ be the highest value such that $\alg(u)\geq \alpha\opt(u)$ for every instance $u\in \underset{\nu\in\mc{V}}{\bigcup}supp(\mc{D}_\nu)$.  
	Plugging this into the equality above proves the claim,
	\[\frac{\mbb{E}_{u\sim \mc{D}_\nu}[\alg(u)]}{\mbb{E}_{u\sim \mc{D}_\nu}[\opt(u)]}\geq \frac{\mbb{E}_{u\sim \mc{D}_\nu}[\alpha\opt(u)]}{\mbb{E}_{u\sim \mc{D}_\nu}[\opt(u)]}=\alpha\qquad \forall \nu\in\mc{V}.\]
	\hfill
\end{proof}
\begin{corollary}\label{metacor}
	When each $\mc{D}_\nu$ is supported only on settings with deterministic rewards, we have,
	\[\inf_{\nu\in\mc{V}}\frac{\alg(\nu)}{\opt(\nu)}\geq (1-1/e).\]
\end{corollary}
\noindent The corollary follows from Lemma \ref{fact1}. Observe that we could replace \alg\ with another algorithm for stochastic rewards that has competitive ratio $\alpha$ in the deterministic setting. The theorem gives us a sufficient condition for obtaining this $\alpha$ competitiveness in a stochastic setting without any new analysis. 
The following general notion will be useful in the subsequent discussion.

For any instance $\nu$, the trivial distribution that has a probability mass of 1 on $\nu$ is value preserving. To apply Corollary \ref{metacor}, we need to show the existence of value preserving distributions that have support only on instances of the deterministic rewards problem. Next, we establish that such distributions exist for every $\nu\in \mc{V}$ when $\mc{V}$ is one of several special cases of the general stochastic rewards problem. 
\subsection{Identical Probabilities ($\mb{p_{it}=p\,\, \forall (i,t)\in E}$)} 

\begin{lemma}\label{identical}
	There exists a value preserving distribution over instances of deterministic rewards for each instance of the stochastic rewards problem with identical probabilities.
\end{lemma}
\begin{proof}
	Let $\nu$ be an arbitrary instance of stochastic rewards with graph $G=(I,T,E)$, vertex weights $\{r_i\}_{i\in I}$, and identical probabilities $p$. 
	To generate a distribution $\mc{D}_\nu$ over instances with deterministic rewards, we independently sample Bernoulli random variables $\{s_t\}_{t\in T}$. The value of $s_t$ is 1 with probability (w.p.) $p$ and 0 w.p.\ $1-p$. We generate an instance with deterministic rewards by selecting the subset of arrivals where $s_t$ is 1. Formally, consider an instance of the deterministic rewards problem with the same set of resources $I$, reduced set of arrivals \[T^*=\{t\mid s_t=1, t\in T\}\] 
	that arrive in the same order, and reduced set of edges
	\[E^*=\{(i,t) \mid i\in I, t\in T^*, (i,t)\in E\}.\]
	Through the above process, the randomness of $\{s_t\}_{t\in T}$ induces a distribution over instances with deterministic rewards. 
	It remains to show that this distribution is value preserving. 
	
	Let $\mc{A}$ be some non-anticipative algorithm. For a randomized algorithm $\mathcal{A}$, arbitrarily fix the (intrinsic) random seeds of the algorithm. We can generate sample paths of $\nu$ using the values of $\{s_t\}_{t\in T}$ as follows. Consider the execution of $\mathcal{A}$ on $\nu$ such that when $\mc{A}$ matches arrival $t\in T$ to some available resource, the match succeeds if and only if $s_t=1$ (an event that occurs with probability $p$). The total reward of $\mc{A}$ on this sample path of $\nu$ is the same as the total reward of $\mc{A}$ on the deterministic reward instance with reduced arrival set $T^*$ defined above. Taking expectation over the binary random variables $\{s_t\}_{t\in T}$ and the (independent) intrinsic randomness of $\mc{A}$ gives us the desired equivalence of expected total reward. 
\end{proof}

\textbf{Remarks:} In the proof above, we could generate sample paths of $\mc{A}$ (on $\nu$) by sampling a single random variable $s_t$ for each arrival because (1) $\mc{A}$ matches an arrival to at most one resource and (2) for identical probabilities, $s_t$ can represent success/failure of matching $t$ regardless of the resource $t$ is matched to.  It is easy to verify that the result generalizes to instances of the stochastic rewards problem where all edges incident on a particular arrival have the same probability, i.e., $p_{it}=p_t\,\, \forall t\in T$. 
In the next section, we show a stronger generalization for decomposable probabilities. 

\subsection{Decomposable Probabilities ($\mb{p_{it}=p_i\,p_t\,\, \forall (i,t)\in E}$)} \hfill\\

\noindent The idea for identical probabilities does not directly apply when the success probabilities for edges incident on a given arrival are different. To reduce the problem to a deterministic setting we first need some intermediate results. 
\begin{lemma}\label{interim}
	For any given instance of stochastic rewards with decomposable probabilities, there exists a value preserving distribution over instances 
	of stochastic rewards where all edges incident on the same resource have identical success probability. 
\end{lemma}
\begin{proof}
	Let $\nu$ be an arbitrary instance of stochastic rewards with graph $G=(I,T,E)$, vertex weights $\{r_i\}_{i\in I}$, and decomposable probabilities. For an algorithm $\mathcal{A}$ that matches each arrival at most once, we can treat the Bernoulli random variable $s_{it}$ that indicates success/failure of an edge $(i,t)$ as a product of two independent Bernoulli random variables, i.e., 
	\[s_{it}=s^t_i\, s_t,\] 
	here $s^t_i$ is 1 w.p.\ $p_i$, $s_t$ is 1 w.p.\ $p_t$. Note that random variable $s_t$ correlates the stochastic rewards of edges incident on $t$ but this is without loss of generality (w.l.o.g.), for any algorithm that matches $t$ to at most one resource. 
	Now, consider the Bernoulli random variables $\{s_t\}_{t\in T}$. 
	Given a realization of these random variables, by removing all arrivals where $s_t=0$, we have an instance $u$ of stochastic rewards on the remaining set of arrivals $T^*$ and remaining edge set $E^*$ such that $p_{it}=p_i\,\, \forall (i,t)\in E^*$. The randomness in $\{s_t\}_{t\in T}$ induces a distribution over such instances. Similarly to the proof of Lemma \ref{identical}, this transformation preserves the expected total reward of non-anticipative algorithms. 
\end{proof}

Consider an instance $\nu$ of stochastic rewards with probabilities $p_{it}=p_i\,\, \forall (i,t)\in E$. 
Due to different probabilities for different resources, sampling a single Bernoulli random variable for each arrival does not suffice to obtain a (value preserving) distribution over deterministic instances. Another possibility is to use a single Bernoulli random variable $s_i$ to determine if a match to resource $i$ succeeds. However, a resource may be matched numerous times before a match is successful. Therefore, we need multiple independent samples of $s_i$ to capture the stochastic reward from a single resource. 
In fact, we show that for every $\nu$, there 
exists a value preserving distribution over instances of $b$-matching with deterministic rewards (where resources may be matched multiple times). To show this we first consider an alternative way to evaluate the problem objective. 

\textbf{Alternative way to compute the objective:} Given an instance $\nu^*$ of stochastic rewards (with arbitrary probabilities) and a non-anticipative algorithm $\mc{A}$, for every $(i,t)\in E$, let $X_{it}$ be a Bernoulli random variable that is 1 if $\mc{A}$ matches $i$ to $t$ and the match succeeds. Let $Y_{it}$ be a Bernoulli random variable that is 1 if $(i,t)$ is included in the matching (the match may fail).  Let $\mbb{E}[X_{it}]$ and $\mbb{E}[Y_{it}]$ denote the expected value of these random variables, here the expectation is with respect to stochastic rewards as well as any intrinsic randomness of $\mc{A}$. For a non-anticipative algorithm,
\[\mbb{E}[X_{it}]=p_{it}\mbb{E}[Y_{it}].\]
By linearity of expectation, we have
\begin{equation}\label{alternate}
	\mc{A}(\nu^*)=\sum_{(i,t)\in E} r_i\mbb{E}[X_{it}]=\sum_{(i,t)\in E} r_ip_{it}\mbb{E}[Y_{it}].
\end{equation}
The RHS above gives an alternate way to evaluate the objective where \emph{we get a reward equal to the expected reward of the edge whenever we make a match (regardless of success/failure)}. In this interpretation, the rewards are deterministic but the successes/failures of resources are still stochastic.

%

In the following lemma, we use a generalization of the notion of value preserving distributions over instances of $b$-matching with unknown budgets. The notion of value preservation is well defined only for algorithms that work for the stochastic rewards problem as well as the $b$-matching with unknown budgets problem. 
Note that \alg\ is one such algorithm. 
\begin{lemma}\label{bmatch}
	There exists a value preserving distribution over instances of $b$-matching with unknown budgets for each instance of the stochastic rewards problem where all edges incident on the same resource have identical success probability. 
\end{lemma}
\begin{proof}
	Let $\nu$ be an arbitrary instance of stochastic rewards with graph $G=(I,T,E)$, vertex weights $\{r_i\}_{i\in I}$, and success probability $p_i$ for every edge incident on $i\in I$. Consider independent Bernoulli random variables $\{s_i\}_{i\in I}$ such that $s_i=1$ w.p.\ $p_i$. For every $i$, we \emph{sequentially} draw $m$ ($=|T|$) independent samples from $s_i$ and let $k_i$ be the index of the first sample that has value 1, with $k_i$ set to $m+1$ when all the samples are 0. The probability that $k_i=k$ is $p_i(1-p_i)^{k-1}$. 
	
	Using the values of these random variables we can generate a sample of $\nu$ as follows. For any given resource $i$, when $i$ is matched for the $k$-th time, the match succeeds if the $k$-th independent sample of $s_i$ equals 1 and fails otherwise. On this sample path, $i$ is successfully matched if it is matched to $k_i$ arrivals. Given a realization of $\{k_i\}_{i\in I}$, consider the instance of $b$-matching with resources $I$, arrivals $T$ (in the same order), edge set $E$, vertex weights $r_ip_i\,\, \forall i\in I$ and budgets $b_i=k_i\,\, \forall i\in I$. The randomness in $\{k_i\}_{i\in I}$ makes the budgets initially unknown to an algorithm and induces a distribution over instances of the $b$-matching problem. 
	Using the alternative way to compute the objective for stochastic rewards (recall \eqref{alternate}), we have that the transformation is value preserving.
	
		\hfill\end{proof}

\begin{theorem}\label{decomp}
	\alg\ is $(1-1/e)$ competitive for the setting of stochastic rewards with decomposable probabilities.
\end{theorem}
\begin{proof}
	Using Theorem \ref{meta}, it suffices to analyze the competitive ratio of \alg\ on instances of stochastic rewards with edge probabilities $p_{it}=p_i\,\, \forall (i,t)\in E$. We claim that the competitive ration of \alg\ on such instances is lower bounded by the competitive ratio of \alg\ on instances of $b$-matching with unknown budgets. The proof is very similar to the proof of Theorem \ref{meta} with a minor change due to the fact that \opt\ for the stochastic rewards instance may not be a well defined algorithm for $b$-matching. Given an instance $\nu$ of stochastic rewards with probabilities $p_{it}=p_i$, from Lemma \ref{bmatch} we have a value preserving distribution $\mc{D}_\nu$ over instances of $b$-matching. Let $u$ denote an instance in $supp(\mc{D}_\nu)$. Let $\opt^*(u)$ denote the optimal offline solution for instance $u$ of the $b$-matching problem when the budgets are apriori known. Since $\opt^*$ has additional information (knows the budgets), we have
	\[\opt(\nu)\leq \mbb{E}_{u\sim \mc{D}_\nu}[\opt^*(u)],\]
	Now,
	\[\frac{\alg(\nu)}{\opt(\nu)}\geq \frac{\mbb{E}_{u\sim \mc{D}_\nu}[\alg(u)]}{\mbb{E}_{u\sim \mc{D}_\nu}[\opt^*(u)]} \qquad \forall \nu\in\mc{V}.\]
	From Lemma \ref{fact2}, we have that $\alg(u)\geq (1-1/e)\opt^*(u)$ for every instance $u$ of the $b$-matching with unknown budgets problem. Plugging this into the equality above proves the claim,
	\[\frac{\mbb{E}_{u\sim \mc{D}_\nu}[\alg(u)]}{\mbb{E}_{u\sim \mc{D}_\nu}[\opt^*(u)]}\geq \frac{\mbb{E}_{u\sim \mc{D}_\nu}[(1-1/e)\opt^*(u)]}{\mbb{E}_{u\sim \mc{D}_\nu}[\opt^*(u)]}=(1-1/e) \qquad \forall \nu\in\mc{V}.\]
\hfill\end{proof}

\subsection{Correlated Stochastic Rewards} \hfill\\

\noindent Our discussion so far has focused on independent stochastic rewards. We now consider a more general setting where rewards are not necessarily independent (can be correlated) and there is a joint distribution over the rewards. The joint distribution is known to offline but online algorithm only knows the marginal success probabilities $p_{it}\,\, \forall (i,t)\in E$. To the best of our knowledge, no previous result is known for the problem under any (non-trivial) form of correlation. The simplicity of our approach reveals that \alg\ is $(1-1/e)$ competitive for decomposable (marginal) edge probabilities even when the rewards have certain types of correlations. 

\textbf{Time correlated rewards:} Recall that when the rewards are independent and edge probabilities are decomposable, we can view the Bernoulli random variable $s_{it}$ that indicates success/failure of an edge $(i,t)$ as a product of two independent Bernoulli random variables, i.e., 
\[s_{it}=s^t_i\, s_t,\] 
here $s^t_i$ is 1 w.p.\ $p_i$, $s_t$ is 1 w.p.\ $p_t$. Random variable $s_t$ correlates the stochastic rewards of edges incident on $t$ but this is w.l.o.g., for any algorithm that matches $t$ to at most once resource.  With this viewpoint in mind, consider a more general setting where the random variables $\{s_t\}_{t\in T}$ are not necessarily independent and their values come from an arbitrary joint distribution $\mc{S}_T$ over $\{0,1\}^{m}$. 
We call this the setting of time correlated rewards. 
\begin{theorem}
	For the setting of time correlated rewards, \alg\ is $(1-1/e)$ competitive against the non-anticipative offline benchmark that knows the joint distribution $\mc{S}_T$.
\end{theorem}

Similar to the proof of Theorem \ref{decomp}, the result above follows from the existence of a value preserving distribution over instances of stochastic rewards with probabilities $p_{it}=p_i\,\, \forall (i,t)\in E$. 
%
\begin{lemma}
	For each instance of time correlated rewards there exists a value preserving distribution over instances of stochastic rewards with probabilities $p_{it}=p_i\,\, \forall (i,t)\in E$.
\end{lemma}
\begin{proof}
	Let $\nu$ denote an instance of the time correlated rewards problem with graph $G=(I,T,E)$. 
	Similar to the proof of Lemma \ref{identical}, consider a random sample of values of $\{s_t\}_{t\in T}$ drawn according to distribution $\mc{S}_T$. Given a non-anticipative algorithm $\mc{A}$, we use this sample to generate a (partial) sample path of $\nu$ as follows. When $\mc{A}$ includes an edge $(i,t)$ in the matching, the edge succeeds if $s_t=1$ and fails otherwise. This is a partial sample path of $\nu$ since $\{s^t_i\}_{i\in I, t\in T}$ are still random. By ignoring the arrivals where $s_t=0$, we have an instance of stochastic rewards with a reduced set of arrivals $T^*$, reduced edge set $E^*$, and edge probabilities $p_{it}=p_i\,\, \forall (i,t)\in E^*$. Taking expectation over the randomness in $\{s^t_i\}_{i\in I, t\in T}$, the expected total reward of $\mc{A}$ on this instance is the same as the expected total reward of $\mc{A}$ on the partial sample path of $\nu$.
\end{proof}

\subsection{Upper Bound for General Probabilities}
\begin{theorem}\label{upper}
	\alg\ is at most 0.624 competitive $(<1-1/e)$ for the general setting of online matching with stochastic rewards.
\end{theorem}
To construct an upper bound we give a hard instance of the problem where the expected objective value of \alg\ is strictly less than $(1-1/e)$. Our instance has small and large edge probabilties, and heterogeneous vertex weights.  

\textbf{Instance of Stochastic Rewards:} We have $n+1$ resources and $2n$ arrivals. Resources 1 through $n$ have an edge to all arrivals. Resource $n+1$ has edges only to the first $n$ arrivals. The probabilities of resource $n+1$'s edge to arrival $t$ is $p\,w(t)$ where $w(t)$ is a decreasing function of $t$.
\[w(t)=\frac{1-e^{\frac{t}{n+1}-1}}{1-e^{-1+\epsilon}},\]
here $\epsilon=0.133$ and $p$ is small enough so that $pw(1)<1$. The weight of vertex $n+1$ is $1/p$. Resources 1 though $n$ have a vertex weight of 1 and all edges incident on them have probability 1. 

This instance is inspired by the following hard instance for the related problem of Adwords with Unknown Budgets. 

\textbf{Instance of Adwords with Unknown Budgets }\cite{unknown}: We have $n+1$ resources and $2n$ arrivals. Resources 1 through $n$ have an edge to all arrivals. Resource $n+1$ has edges to the first $n$ arrivals only. Arrival $t\in[n]$ bids 
\[w(t)=\frac{1-e^{\frac{t}{n+1}-1}}{1-e^{-1+\epsilon}},\]
(here $\epsilon=0.133$) on resource $n+1$. Resource $n+1$ has a total budget of $\sum_{t\in[n]} w(t)$. Resources 1 though $n$ have a budget of 1 and every arrival bids 1 on each of these resources. 


Let $\nu$ refer to the instance of stochastic rewards above and let $\nu^*$ refer to the instance of Adwords. Let $\alg^*$ denote the Perturbed Greedy algorithm for Adwords. This algorithm greedily matches each arrival to an available neighboring resource with the highest perturbed bid, here $b_{it} (1-e^{y_i-1})$ is the perturbed bid of resource $i$ on arrival $t$. In other words, $\alg^*$ is similar to $\alg$ with the edge probabilities replaced by bids. Resources become unavailable in $\alg^*$ when the budget runs out. Let $\alg(\nu)$ and $\alg^*(\nu^*)$ refer to the expected total reward of \alg\ on instance $\nu$ and $\alg^*$ on instance $\nu^*$ respectively.  Let $\opt^*(\nu^*)$ denote the optimal value of the offline solution for the Adwords instance $\nu^*$ and let $\opt(\nu)$ refer to the optimal expected value of the non-anticipative offline solution for $\nu$. We use the notation $O(np)$ to simplify notation and ignore absolute constants that do not depend on $n$ and $p$.
\begin{lemma}\label{gap}
	For the instances described above, when $p\ll \frac{1}{n}$, we have $1\leq \frac{\alg^*(\nu^*)}{\alg(\nu)}\leq (1+O(np))$ and $1\leq \frac{\opt^*(\nu^*)}{\opt(\nu)}\leq (1+O(np))$.
\end{lemma}
\begin{proof}
	Consider a parallel execution of \alg\ on $\nu$ and $\alg^*$ on $\nu^*$ with the same values for the random variables $\{y_i\}_{i\in I}$. In the following, assume that resource $n+1$ is never matched successfully in $\alg$. Then, both algorithms make identical decisions, i.e., if $\alg$ matches $t$ to resource with index $i$ then so does $\alg^*$. Suppose that $\alg$ (and $\alg^*$) match arrival $t$ to resource $i$. Using the alternative way of computing total reward for $\nu$ (recall equation \eqref{alternate}), we have that \alg\ gets a deterministic reward equal to $1$ if $i\in[n]$ and a reward of $r_i\,p_{it}=w(t)$ if $i=n+1$. $\alg^*$ gets the same reward at $t$ and the two algorithms have identical total reward. Similarly, the optimal offline rewards are identical when resource $n+1$ is not successfully matched by the offline algorithm for $\nu$. 
	
	On sample paths where resource $n+1$ is successfully matched in $\nu$, $\alg^*$ and $\opt^*$ may have a higher total reward than their counterparts. Thus, $\alg^*(\nu^*)\geq \alg(\nu)$ and $\opt^*(\nu^*)\geq \opt(\nu)$. The probability that resource $n+1$ is successfully matched by any algorithm is at most $1-(1-p)^n$. For $p\ll \frac{1}{n}$, we have $(1-p)^n= 1-O(np)$. Notice that $\alg(\nu)\geq n$ and $\opt(\nu)\geq n$ since both have total reward at least $n$ on every sample path (resources 1 to $n$ are always matched). Further, the maximum possible reward in either setting is at most $n+\sum_{t\in[n]} w(t)=O(n)$. Therefore, 
	\[\alg(\nu)\leq \alg^*(\nu^*)\leq\alg(\nu)+(1-(1-p)^n)O(n)=(1+O(np))\alg(\nu).\] Similarly, $\opt(\nu)\leq \opt^*(\nu^*)\leq \opt(\nu)+(1-(1-p)^n)O(n)=(1+O(np))\opt(\nu)$.
\end{proof}
\begin{theorem}[Theorem 5 in \cite{unknown} (restated).]\label{unk}
	For the instance $\nu^*$ of the Adwords problem, we have $\lim_{n\to+\infty}\frac{\alg^*(\nu^*)}{\opt^*(\nu^*)}\leq 0.624$.
\end{theorem}
\begin{proof}[Proof of Theorem \ref{upper}.]
	From Lemma \ref{gap} and Theorem \ref{unk}, we get
	\[\frac{\alg(\nu)}{\opt(\nu)}=\frac{\alg(\nu)}{\alg^*(\nu^*)}\frac{\opt^*(\nu^*)}{\opt(\nu)}\frac{\alg^*(\nu^*)}{\opt^*(\nu^*)}\leq 1\times (1+O(np))\times \frac{\alg^*(\nu^*)}{\opt^*(\nu^*)}.\]
	For $p=\frac{1}{n^2}$, we have, 
	$\lim_{n\to+\infty} \frac{\alg(\nu)}{\opt(\nu)}\leq 0.624,$ as desired.
\end{proof}
\section{Conclusion}
We consider the problem of online matching with stochastic rewards and show that from a competitive ratio standpoint the setting reduces to the simpler deterministic online matching when the probabilities are decomposable. The simplicity of our approach reveals a more general (tight) $(1-1/e)$ competitive ratio for the Perturbed Greedy algorithm in a setting where edge probabilities are decomposable and the arrival associated randomness ($s_t$) in the stochastic rewards may be arbitrarily correlated across time. Finally, we show that for the general stochastic rewards problem, Perturbed Greedy is strictly less than $(1-1/e)$ competitive. The key open question for this problem is to find an algorithm with the highest possible competitive ratio for the general setting. Our work also opens up the study of the even more general setting of correlated stochastic rewards, for which we give the first non-trivial competitive ratio results.

\bibliographystyle{alpha} 
\bibliography{new}
\end{document}